\documentclass[preprint,12pt]{elsarticle}

\usepackage{amssymb}
\usepackage{amsmath}

\usepackage{amssymb,amsmath,amsfonts,amscd,amsthm,amsbsy}
\usepackage{braket,algorithm,algorithmic}
\usepackage{upgreek,comment,color}
\usepackage{bbm} 
\usepackage{enumitem}
\usepackage[colorlinks,citecolor=black,linkcolor=black,urlcolor=blue]{hyperref} 
\usepackage{subcaption}
\usepackage{booktabs}
\usepackage{graphics}

\newtheorem{corollary}{Corollary}
\newtheorem{proposition}{Proposition}

\newcommand{\given}{\;|\;}

\newcommand{\ma}{\mathbf{A}}
\newcommand{\mb}{\mathbf{B}}

\newcommand{\md}{\mathbf{D}}

\newcommand{\mg}{\mathbf{G}}
\newcommand{\mh}{\mathbf{H}}
\newcommand{\mi}{\mathbf{I}}

\newcommand{\mr}{\mathbf{R}}
\newcommand{\ms}{\mathbf{S}}

\newcommand{\mU}{\mathbf{U}}
\newcommand{\mv}{\mathbf{V}}
\newcommand{\mw}{\mathbf{W}}
\newcommand{\mx}{\mathbf{X}}

\newcommand{\my}{\mathbf{Y}}

\newcommand{\sigmaminplus}{\sigma^{+}_{\min}}

\newcommand{\mwln}{\mw_{\text{LN}}}
\newcommand{\mwls}{\mw_{\text{LS}}}

\newcommand{\vc}{\mathbf{c}}

\newcommand{\vr}{\mathbf{r}}

\newcommand{\vu}{\mathbf{u}}
\newcommand{\vv}{\mathbf{v}}
\newcommand{\vw}{\mathbf{w}}
\newcommand{\vx}{\mathbf{x}}

\newcommand{\vy}{\mathbf{y}}

\DeclareMathOperator{\trace}{\mathrm{trace}}

\DeclareMathOperator{\Exp}{\mathbb{E}}
\newcommand{\Tra}{^{\sf T}} 
\newcommand{\Inv}{^{-1}} 
\newcommand{\Dag}{^{\dagger}} 

\DeclareMathOperator{\rank}{\mathrm{rank}}

\newcommand{\vone}{{\bf 1}}

\newcommand{\mPxi}{\mathbf{P}_{\vx_{i}}}
\def\span{\mathrm{span}} 
\newcommand{\vgamma}{\boldsymbol{\gamma}}
\newcommand{\vtheta}{\boldsymbol{\theta}}

\begin{document}

\begin{frontmatter}


\author{Jocelyn T. Chi\corref{cor1}\fnref{label2}}
\ead{jocelync@umn.edu}
\affiliation{organization={University of Minnesota Twin Cities},
             city={Minneapolis},
             postcode={55455},
             state={MN},
             country={USA}}

\title{Large Scale High-Dimensional Reduced-Rank Linear Discriminant Analysis}




\begin{abstract}
Reduced-rank linear discriminant analysis (RRLDA) is a foundational method of dimension reduction for classification that has been useful in a wide range of applications.  The goal is to identify an optimal subspace to project the observations onto that simultaneously maximizes between-group variation while minimizing within-group differences.  The solution is straight forward when the number of observations is greater than the number of features but computational difficulties arise in both the high-dimensional setting, where there are more features than there are observations, and when the data are very large.  Many works have proposed solutions for the high-dimensional setting and frequently involve additional assumptions or tuning parameters.  We propose a fast and simple iterative algorithm for both classical and high-dimensional RRLDA on large data that is free from these additional requirements and that comes with guarantees.  We also explain how RRLDA-RK provides implicit regularization towards the least norm solution without explicitly incorporating penalties.  We demonstrate our algorithm on real data and highlight some results.
\end{abstract}




\end{frontmatter}




\section{Introduction}

Reduced-rank linear discriminant analysis \cite[Section 4.3.3]{esl} (RRLDA) is a foundational method of dimension reduction for downstream classification \cite{duda2000pattern, fukunaga1993statistical, esl}.  RRLDA has many names in the literature and is also widely known as classical linear discriminant analysis (LDA),  multiclass Fisher's LDA (FLDA), and multiple discriminant analysis.  The goal of RRLDA is to identify an optimal subspace for projecting observations onto that simultaneously maximizes between-group variation while minimizing within-group differences.  After this transformation, the reduced dimension data are well-suited for downstream classification tasks.  RRLDA and its variants have been a popular tool for dimension reduction for classification in a wide range of applications.  These include facial recognition \cite{liu2002gabor, liu1998enhanced, liu2004improving, huang2012linear}, hyperspectral imaging \cite{bandos2009classification, du2007modified, du2001linear}, speech and music classification  \cite{alexandre2005application}, medical diagnoses \cite{chen2011new, liu2020near}, magnetic resonance imaging classification \cite{ni2020prediction, witjes2003multispectral, lin2010automated}, and human detection in thermal video imaging \cite{sharma2016fisher}.

The RRLDA subspace is computationally straightforward to obtain when the number of observations is greater than the number of features but computational difficulties arise in two situations.  The first is the high-dimensional setting, where there are many more features than there are observations.  The second is the very large data situation, in which both the number of observations and the number of features are very large.  Many works propose solutions for the high-dimensional setting and frequently involve additional assumptions or tuning parameters.

We propose a scalable randomized iterative algorithm for performing RRLDA on large and high-dimensional data that is free from simplifying assumptions and that also comes with guarantees.  We discuss some intuition for why this algorithm might work on high-dimensional data without explicit regularization.  We also demonstrate the utility of this algorithm on real large and high-dimensional datasets on a number of downstream classification tasks.

\subsubsection{Overview}

In the rest of this section, we first set some notation and abbreviations.  In Section \ref{sec:rrlda}, we briefly introduce the RRLDA problem and discuss related works and computational issues in the high-dimensional data regime.  In Section \ref{sec:lsLDA}, we present the least squares formulation of RRLDA.  In Section \ref{sec:rk}, we briefly review the randomized Kaczmarz method and then apply it to the least squares formulation of RRLDA in Section \ref{sec:rrlda-rk}.  We also present convergence analysis and intuition on how the algorithm imparts implicit regularization towards the least norm solution.  Finally, we conclude with numerical experiments on real datasets in Section \ref{sec:experiments} and a brief discussion in Section \ref{sec:conclusion}.

\subsection{Notation and Abbreviations}

We define matrices by boldface capital letters $\mx$, vectors by boldface lowercase letters $\vx$, and scalars by lowercase letters $y$.  We assume that all vectors are column vectors.  We denote the Euclidean vector norm of a vector $\vx$ by $\| \vx \|$ and the Frobenius matrix norm of a matrix $\mx$ by $\|\mx\|_{\text{F}}$.  We denote the $i^{th}$ row of a matrix $\mx$ by $\vx_{i}\Tra$.  Let $\sigmaminplus(\mx)$ denote the smallest nonzero singular value of $\mx$. Let $\mi_{d}$ denote the $d \times d$ identity matrix.  Let $\mathcal{R}(\mx)$ and $\mathcal{N}(\mx)$ denote the range and null spaces of $\mx$, respectively.  We denote the inner product between two matrices $\ma$ and $\mb$ of conforming dimensions as $\langle \ma, \mb\rangle = \trace\{\ma\Tra\mb\}$.  We employ $\mathbb{N}^{+}$ to denote the positive integers and $[n]$ to denote the set of integers $\{1, 2, \dots, n\}$.  We abbreviate the singular value decomposition of a matrix as its SVD and the randomized Kaczmarz method as RK.

\section{Reduced-Rank Linear Discriminant Analysis}
\label{sec:rrlda}

Suppose we have a dataset containing $n$ labeled observations, each belonging to one of $g$ classes.  Following the notation in \cite{ye2007least}, let $\{(\vx_{i}, y_{i})\}$ for $i \in [n]$ denote the $i^{th}$ observation, where $\vx_{i} \in \mathbb{R}^{d}$ contains $d$ measurements and $y_{i} \in [g]$ denotes the corresponding known class label.  Let $\mx \in \mathbb{R}^{n \times d}$ denote the matrix comprised of the $\vx_{i}\Tra$ on the rows after column centering.  Let $n_{j}$ denote the number of observations in the $j^{th}$ class with $\sum_{j=1}^{g} n_{j}= n$ since each observation belongs to exactly one of $g$ classes.  Let $\vc^{(j)} \in \mathbb{R}^{d}$ denote the $j^{th}$ class centroid, or mean, and let $\vc \in \mathbb{R}^{d}$ denote the grand centroid of all observations.

Originally introduced in \cite{fisherlda}, the goal of RRLDA is to find a linear transformation matrix $\mg \in \mathbb{R}^{d \times \ell}$, with $\ell < d$, such that the transformed class centroids $\mg\Tra\vc^{(j)} \in \mathbb{R}^{\ell}$ are maximally separated with respect to variance.  Therefore, the optimal transformation $\mg$ simultaneously maximizes the variance between classes while minimizing the variance within classes.

To see how to obtain this optimal transformation $\mg$, we define the \emph{within-class}, \emph{between-class}, and \emph{total scatter matrices} \cite[Section 4.11]{duda2000pattern}
\begin{eqnarray*}
	\ms_{w} &=& \frac{1}{n} \sum_{j=1}^{g} \sum_{i:\, y_{i} = j} (\vx_{i} - \vc^{(j)})(\vx_{i} - \vc^{(j)})\Tra \in \mathbb{R}^{d \times d},\\
	\ms_{b} &=& \frac{1}{n} \sum_{j=1}^{g} n_{j}(\vc^{(j)} - \vc)(\vc^{(j)} - \vc)\Tra \in \mathbb{R}^{d \times d}, \text{ and} \\
	\ms_{t} &=& \frac{1}{n} \sum_{i=1}^{n} (\vx_{i} - \vc)(\vx_{i} - \vc)\Tra \in \mathbb{R}^{d \times d},
\end{eqnarray*}
where $\ms_{t} = \ms_{w} + \ms_{b}$.  There are multiple approaches to summarizing the amount of scatter in these matrices (\cite[Equation (106)]{duda2000pattern}, \cite{ye2007least}).  Note that $\trace(\ms_{w})$ and $\trace(\ms_{b})$ can be written as \cite{ye2007least}
\begin{eqnarray*}
	\trace(\ms_{w}) &=& \frac{1}{n} \sum_{j=1}^{g} \sum_{y_{i} = j} \|\vx_{i} - \vc^{(j)}\|^{2} \quad \text{ and} \quad
	\trace(\ms_{b}) = \frac{1}{n} \sum_{j=1}^{g} n_{j} \|\vc^{(j)} - \vc\|^{2}
\end{eqnarray*}
to highlight how the terms capture the variation among observations within a class and the differences between the classes, respectively.

In the transformed space, the observations are $\mg\Tra\vx_{i} \in \mathbb{R}^{\ell}$ for $i \in [n]$ and the corresponding scatter matrices are
\begin{eqnarray*}
	\ms_{w}^{G} &=& \mg\Tra\ms_{w}\mg \in \mathbb{R}^{\ell \times \ell},\\
	\ms_{b}^{G} &=& \mg\Tra\ms_{b}\mg \in \mathbb{R}^{\ell \times \ell}, \text{ and} \\
	\ms_{t}^{G} &=& \mg\Tra\ms_{t}\mg \in \mathbb{R}^{\ell \times \ell}.
\end{eqnarray*}
Since ${\ms_{t}^{G} = \ms_{w}^{G} + \ms_{b}^{G}}$, an optimal transformation $\mg^{\star}$ that maximizes differences between classes while minimizing variation within classes can be obtained by solving 
\begin{eqnarray}
	\mg^{\star} &=& \arg \max_{\mg} \left\{ \trace\left[ (\ms_{t}^{G})\Inv \ms_{b}^{G} \right] \right\}
	\label{eqn:ulda1}
\end{eqnarray}
if $(\ms_{t}^{G})\Inv$ exists.  This $\mg^{\star}$ corresponds to the eigenvectors of the nonzero eigenvalues of the following generalized eigenvalue problem \cite[Section 4.11, Equation (104)]{duda2000pattern}
\begin{eqnarray}\label{eqn:geneigen}
    \ms_{b} \vw &=& \lambda \ms_{t} \vw.
\end{eqnarray}
Since $\rank(\ms_{b}) \le g-1$, there are at most $g-1$ non-zero eigenvalues in \eqref{eqn:geneigen}.  Consequently, $\mg^{\star}$ has at most $g-1$ columns and therefore, $\ell$ is typically set to be $g-1$ \cite[Section 4.11]{duda2000pattern}.

\subsection{Related works on classical linear discriminant analysis}
\label{sec:relatedworks}

A variant of RRLDA is Gaussian-model LDA (GLDA) \cite[Section 4.3]{esl}.  Unlike RRLDA, which makes no distributional assumptions on the data, GLDA assumes that observations come from class-conditional Gaussian distributions.  A least squares formulation exists for binary-class GLDA and \cite{chi2025linear} describes and analyzes a scalable, randomized algorithm for it.  Random projections approaches to Gaussian-model LDA also exist \cite{durrant2010compressed, elkhalil2019asymptotic}.  

Computational algorithms for RRLDA include those based on a least squares formulation \cite{cai2008training, ye2007least} and deep learning approaches \cite{dorfer2015deep, diaz2019deep}.  We refer to the $n\ge d$ and $n < d$ cases as \emph{classical} and \emph{high-dimensional} RRLDA, respectively.  We next discuss the many works that aim to address difficulties in high-dimensional RRLDA.

\subsection{High-dimensional scenarios}
\label{sec:high-dimensional}

In many applications, the number of variables $d$ may be larger than the number of observations $n$.  For example, this is common in imaging datasets, where $d$ denotes the number of pixels and $n$ denotes the number of images.  In these cases, $\ms_{t}$ is singular.  

Many adaptations for high-dimensional RRLDA attempt to address the dilemma of singular scatter matrices, or singular covariance matrices in the case of GLDA.  However, these frequently involve additional assumptions or tuning parameters.  For example, in the Gaussian-model case, many works make simplifying assumptions on the covariance matrix, and are sometimes also described as regularization (towards non-singular or other structure) on the covariance matrices.  When no Gaussian modeling assumptions are made, the corresponding regularization can occur on scatter matrices.  Examples include diagonal \cite{friedman1989regularized, dudoit2002comparison, bickel2004some}, positive definite \cite{ye2006efficient, krzanowski1995discriminant, guo2007regularized}, and spiked \cite{sifaou2020high} covariance or scatter models.

Another common adaptation is to incorporate penalties involving tuning parameters.  For example, some works employ regularization on the discriminant vectors.  These include the $\ell_{1}$ and fused Lasso penalties \cite{witten2011penalized} and the $\ell_{1}$ penalty using a least squares framework for sparse RRLDA \cite{mai2012direct}.  Some apply the $\ell_{1}$ penalty to the ratio of the between-class to within-class scatter terms within the objective function \cite{wang2013fisher}.

Beyond modifications to adapt RRLDA for high-dimensional scenarios, another common approach is to simply perform dimension reduction on the column space prior to dimension reduction with RRLDA.  One example is to perform principal component analysis (PCA) prior to RRLDA \cite{belhumeur1997eigenfaces, swets2002using, zhao2008subspace}.  However, SVD-based approaches such as PCA require $\mathcal{O}(n^{2}d)$ flops in high-dimensional cases, and may be computationally prohibitive for very large data.  Since these approaches essentially preprocess the data prior to performing classical RRLDA, they are not a focus of this work. 

In lieu of simplifying assumptions or tuning parameters, one natural alternative is to employ the Moore-Penrose generalized inverse $\ms_{t}\Dag$, in lieu of $\ms_{t}\Inv$ \cite{raudys1998expected}.  One example is Uncorrelated LDA (ULDA) \cite[Algorithm 1]{ye2005characterization}, where we obtain
\begin{eqnarray}
	\mg^{U} &=& \arg \max_{\mg} \left\{ \trace\left( \ms_{b}^{G} (\ms_{t}^{G})\Dag \right) \right\}.
	\label{eqn:ulda}
\end{eqnarray}
Then $\mg^{U}$ corresponds to the eigenvectors of $\ms_{t}\Dag\ms_{b}$ that have non-zero eigenvalues \cite{ye2005characterization}.  

\subsection{Other related works on randomized high-dimensional RRLDA}

A randomized iterative algorithm for regularized RRLDA based on iterative randomized matrix multiplication, or sketching, also exists \cite{chowdhury2019randomized}.  This approach \cite[Algorithm 1]{chowdhury2019randomized} involves computing an SVD of $\mx\ms$, where $\ms\in \mathbb{R}^{d \times s}$ is an appropriately selected random matrix with $s \ll d$ and then performing a sequence of linear solves with it.  By contrast, Algorithm \ref{alg:rkRRLDA} in Section \ref{sec:rrlda-rk} requires no computations for selecting an appropriate $\ms$ (such as computing full or approximate leverage scores), no SVD computations, requires only one row of $\mx$ at a time, makes no structural assumptions on $\mx$, and is free of tuning parameters.

Another sketching-based approach performs dimension reduction prior to RRLDA through matrix multiplication with a random-valued lower dimensional matrix \cite{tu2014making}.  As such, it aligns with previously described approaches that perform dimension reduction as a pre-processing step prior to RRLDA.  By contrast, Algorithm \ref{alg:rkRRLDA} is a randomized algorithm for performing RRLDA that requires no additional computations for selecting $\ms$, no pre-processing procedures, and accesses only one row of $\mx$ at a time.

The method in \cite{jayaprakash2018randomized} also follows the dimension reduction prior to RRLDA approach described previously.  Rather than performing PCA, however, dimension reduction is achieved through generation of a matrix of random Fourier features in $\mathbb{R}^{n \times s}$, where $s \ll d$.  The goal of this approach is to mimic kernel discriminant analysis \cite{baudat2000generalized} without forming an $n \times n$ matrix.

\subsection{Computational Complexity of RRLDA}
\label{sec:compcostRRLDA}

In addition to the computational issues that arise in the high-dimensional case so that $\ms_{t}$ is singular, additional computational constraints may arise in the case that both $n$ and $d$ are very large.  Assuming no additional structure so that $\ms_{b}$ and $\ms_{t}$ are general dense matrices, the computational cost of obtaining the eigenvectors in \eqref{eqn:geneigen} requires $\mathcal{O}(d^{3})$ flops using naive approaches.

Accounting for the design of $\ms_{b}$ and $\ms_{t}$, however, can lead to computational savings.  Following the discussion in \cite[Section I, Part B]{cai2008training}, let $t=\min(n, d)$ and suppose that the total number of classes $g \ll d$.  Then computing $\mg^{\star}$ for RRLDA requires $\mathcal{O}(ndt) + \mathcal{O}(t^{3})$ flops \cite{cai2008training}.  In the high-dimensional case, $t=n$ so this cost becomes $\mathcal{O}(n^{2}d) + \mathcal{O}(n^{3})$ flops.  If $d \gg n$, then we have $\mathcal{O}(n^{2}d)$ flops.  This approach, however, requires computing an SVD of $\mx$.  If both $n$ and $d$ are very large, this can be computationally prohibitive.

\section{Least Squares Formulation of Reduced Rank Linear Discriminant Analysis}
\label{sec:lsLDA}

The solution to \eqref{eqn:ulda1}, and to \eqref{eqn:ulda} in the case that $\ms_{t}$ is singular, can be obtained via a least squares formulation \cite{duda2000pattern, fukunaga1993statistical, esl}.  Following \cite{ye2007least}, let $\my \in \mathbb{R}^{n \times g}$ be the recoded matrix of class labels with 
\begin{eqnarray}
	\my_{ij} &=& \begin{cases}
		\sqrt{\frac{n}{n_{j}}} - \sqrt{\frac{n_{j}}{n}} &\text{if} \quad y_{i} = j, \\
		-\sqrt{\frac{n_{j}}{n}} & \text{otherwise}
	\end{cases}.
	\label{eqn:makey}
\end{eqnarray}
Suppose that both $\mx$ and $\my$ have been column-centered.  Then we solve 
\begin{eqnarray}
	\mw^{\star} &=& \min_{\mw \in \mathbb{R}^{d \times g}}\, \frac{1}{2} \|\mx\mw -\my\|_{\text{F}}^{2}.
	\label{eqn:ls}
\end{eqnarray}
In the high-dimensional case, the system in \eqref{eqn:ls} is underconstrained and consequently, there are either no, or infinitely many, solutions.  A common approach is to seek a least (Frobenius) norm solution \cite[Section 5.6.2]{golub2013matrix} since it is computationally straightforward.  Therefore, $\mw^{\star}$ is the least squares solution ${\mwls = (\mx\Tra\mx)\Dag\mx\Tra\my}$ in the case that $n \ge d$ and least norm solution ${\mwln = (\mx\mx\Tra)\Dag\mx\my = \ms_{t}\Dag\mh_{b}}$ if $n < d$ (\cite[Theorem 4.1]{ye2007least}, \cite[(37)]{ye2007least}).

Notice that $\mw^{\star}$ has $g$ columns while $\mg^{U}$ has $g-1$ columns.  Nonetheless, $\mwln$ and $\mg^{U}$ from \eqref{eqn:ulda} are equivalent in the sense that they result in the same subspace up to padding with zeros and rotation by an orthogonal matrix \cite[Theorem 5.1]{ye2007least} under the assumption that 
\begin{eqnarray*}
	\rank(\ms_{t}) &=& \rank(\ms_{b}) + \rank(\ms_{w}).
\end{eqnarray*}
This rank assumption holds whenever $\mx$ has linearly independent observations \cite[Proposition 5.1]{ye2007least} so that $\rank(\ms_{t}) = n-1$.  This assumption appears to be reasonable for many high-dimensional datasets, holding for nine of ten high-dimensional datasets in \cite[Section 6]{ye2007least}.

\subsection{Computational Complexity of Least Squares RRLDA}
\label{sec:compcostLSLDA}

The computational cost of solving \eqref{eqn:ls} in the classic case where $n > d$ requires $\mathcal{O}(nd^{2})$ flops \cite[Section 5.5.6]{golub2013matrix}.  In the high-dimensional case, it requires $\mathcal{O}(n^{2}d)$ flops.  This in itself offers no improvement over the computational cost of RRLDA as described in Section \ref{sec:compcostRRLDA}.  However, casting RRLDA in the least squares framework enables use of the full spectrum of numerical tools available for solving large-scale linear systems.  Towards a scalable algorithm for \eqref{eqn:ls}, we next overview the randomized Kaczmarz method for solving large systems of linear equations.

\section{The Randomized Kaczmarz Method}
\label{sec:rk}

The Kaczmarz method was introduced in \cite{karczmarz1937angenaherte} and is a deterministic, iterative method for solving linear systems of equations one row at a time.  If the rows are selected at random with probabilities proportional to their row norms, then the randomized Kaczmarz method (RK) \cite{agaskar2014randomized, censor2009note, dai2013randomized, lin2015learning, liu2016accelerated, liu2014asynchronous, needell2010randomized, needell2016stochastic,  niu2020greedy, nutini2016convergence, strohmer2009randomized, zouzias2013rek, jiao2017preasymptotic} converges linearly in expectation to the least squares solution in the overdetermined, consistent case \cite{strohmer2009randomized}.  Convergence results for the overdetermined, noisy case \cite{needell2010randomized}, and the potentially inconsistent least squares case \cite{needell2016stochastic} also exist.

Typically in RK, the right-hand side in \eqref{eqn:ls} is a vector.  Beginning with an initial vector estimate $\vw_{0}$, the RK proceeds as follows.  In the $k^{th}$ RK iterate, it selects the $i^{th}$ row of $\mx$ and its corresponding right-hand side $\vy_{i}$ at random according to some sampling distribution.  Then it computes the $(k+1)^{th}$ update as
\begin{eqnarray*}\label{eqn:rk}
	\vw_{k+1} = \vw_{k} + \frac{y_{i} - \langle \vx_{i}, \vw_{k}\rangle}{\|\vx_{i}\|^{2}_{2}} \, \vx_{i}.
\end{eqnarray*}
Therefore, RK is an instance of weighted stochastic gradient descent where the objective function is the least squares objective \cite{needell2016stochastic}. 

To utilize the RK in the least squares formulation of RRLDA, we employ some modifications.  First, we transform the initial group membership vector $\vy$ into the membership membership matrix $\my$ as described in \eqref{eqn:makey}.  Second, since this results in matrix right-hand sides and solution subspaces in $\eqref{eqn:ls}$, we adapt the update above accordingly.

In the underdetermined and consistent case, RK \cite{strohmer2009randomized}, the randomized extended Kaczmarz \cite{zouzias2013rek}, and the randomized extended Gauss-Seidel \cite{ma2015regs} methods all converge to the least norm solution \cite[Table 1]{ma2015regs}.  Since we hypothesize that very high-dimensional datasets are likely to be consistent, we focus on the RK since it is  computationally the simplest and requires no additional extension procedures.  We explore consistency of real, very high-dimensional datasets in our numerical experiments in Section \ref{sec:experiments}.  In the overdetermined inconsistent case, the RK converges to a radius of the least squares solution \cite{needell2010randomized, needell2016stochastic}.  Our analysis in Section \ref{sec:rrlda-rk} show that this is also true of underdetermined inconsistent systems.  However, this can be rectified by tail averaging the RK iterates after a suitable burn-in period \cite{epperly2024randomized}.  Our numerical experiments in Section \ref{sec:experiments} demonstrate that the RRLDA subspace obtained with the RK performs as well, or very nearly as well, as state-of-the-art methods, on downstream classification tasks even without additional procedures.

\section{Large Scale High-Dimensional Reduced-Rank Linear Discriminant Analysis}
\label{sec:rrlda-rk}

We employ the least squares formulation in \eqref{eqn:ls} and the RK as described in Section \ref{sec:rk} to obtain a scalable algorithm for large-scale RRLDA.  Since the right-hand side in \eqref{eqn:ls} is a matrix, these RK iterates are likewise matrices.  Algorithm \ref{alg:rkRRLDA} depicts the procedures for RRLDA via the RK (\verb|RRLDA-RK|).  

\begin{algorithm}[ht] 
	\caption{Reduced-Rank LDA via the Randomized Kaczmarz Method (RRLDA-RK)}
	\begin{flushleft}
		\vspace{-0.2cm}
	{\bf Input:} Labeled data $(\vy \in \mathbb{R}^{n}$, $\mx \in \mathbb{R}^{n \times d}$)  where observations belong to one of $g$ classes, initial $\mw_{0} \in \mathbb{R}^{d \times g}$ belonging to the range of $\mx\Tra$, maximum iterations $K \in \mathbb{N}^{+}$, sampling probabilities $p_{i}$ for $1\le i \le n$ with $\sum_{i} p_{i} =1$. \\
	{\bf Output:} Reduced dimension subspace $\mw \in \mathbb{R}^{d \times g}$
	\end{flushleft}
	\begin{algorithmic}[1]
			\STATE Form column-centered $\my$ from $\vy$ according to \eqref{eqn:makey}
			\STATE $\mx \gets \mx(\mi - \frac{1}{n}\vone\vone\Tra) $ \COMMENT{Column-center $\mx$}
			\STATE{For $k=0, 1, 2, \dots, K-1$:}
			\STATE $\quad$Randomly sample $i$ from $[n]$ with probability $p_{i}$
			\STATE $\quad \mw_{k+1} \gets \mw_{k} +\, \frac{\vx_{i}}{\|\vx_{i}\|^{2}_{2}} \cdot (\vy_{i}\Tra - \vx_{i}\Tra\mw_{k}) $
			\STATE $\mw \gets \mw^{k+1}$
			\RETURN Reduced dimension subspace $\mw$ 
		\end{algorithmic}
	\label{alg:rkRRLDA}
\end{algorithm}

Proposition \ref{prop:W-inconsistent} presents finite iteration guarantees for the underdetermined (high-dimensional), potentially inconsistent, and matrix regression case (the right-hand side and coefficients are matrices rather than vectors) that is relevant to RRLDA.

\begin{proposition}\label{prop:W-inconsistent}
	Let $\mw^{\star}$ denote the least norm solution of \eqref{eqn:ls}, and suppose that we select the $i^{th}$ observation $\{\vy_{i}, \mx_{i}\}$ with probability $p_{i} = \frac{\|\vx_{i}\|^{2}}{\|\mx\|_{\text{F}}^{2}}$.  Let $\kappa(\mx) = \frac{\|\mx\|^{2}_{\text{F}}}{\{\sigma^{+}_{\min}(\mx)\}^{2}}$ be a scaled condition number, and let $\mr^{\star} = \my - \mx\mw^{\star}$ denote the residual matrix at $\mw^{\star}$.  Then the $k^{th}$ iterate $\mw_{k}$ of Algorithm \ref{alg:rkRRLDA} satisfies
	\begin{eqnarray*}
		\Exp\left[ \|\mw_{k} - \mw^{\star}\|^{2}_{\text{F}} \right]
		&\le& \left(1-\frac{1}{\kappa(\mx)}\right)^{k}\left\|\mw_{0} - \mw^{\star}\right\|^{2}_{\text{F}} + \beta \, \|\mr^{\star}\|^{2}_{\text{F}}.
	\end{eqnarray*}
    where $\beta = \frac{1}{\{\sigma^{+}_{\min}(\mx)\}^{2}}$.
\end{proposition}	

\begin{proof}
   Since the system may be inconsistent, we have $\vy_{i}\Tra = \vx_{i}\Tra \mw^{\star} + \vr_{i}\Tra$ for each row $i$ with $i \in [n]$, where $\vr_{i}\Tra$ is the $i^{th}$ row of $\mr^{\star}$.  We employ this equation involving $\vy_{i}\Tra$, combine terms, and factor out a negative sign to obtain the RK residual at the $k^{th}$ iterate
   \begin{eqnarray*}
       \vy_{i}\Tra - \vx_{i}\Tra \mw_{k} &=& -\vx_{i}\Tra(\mw_{k} - \mw^{\star}) + \vr_{i}\Tra.
   \end{eqnarray*}
   Let $\mPxi = \frac{\vx_{i}\vx_{i}\Tra}{\|\vx_{i}\|^{2}_{2}}$ denote the orthogonal projection matrix onto $\span\{\vx_{i}\}$.  We employ the RRLDA-RK update from Algorithm \ref{alg:rkRRLDA} (line 5), insert the definition of $\vy_{i}$ with respect to the residual at $\mw^{\star}$, factor out a negative sign, and combine terms to obtain
   \begin{eqnarray*}
       \mw_{k+1} - \mw^{\star} &=& (\mi - \mPxi)(\mw_{k} - \mw^{\star}) + \frac{\vx_{i}}{\|\vx_{i}\|^{2}}\vr_{i}\Tra.
   \end{eqnarray*}
   Taking the squared Frobenius norm on both sides, expanding the squared norm on the right, and combining terms gives
   \begin{eqnarray*}
       \|\mw_{k+1} - \mw^{\star}\|^{2}_{\text{F}} 
       &=& \|(\mi - \mPxi)(\mw_{k} - \mw^{\star})\|^{2}_{\text{F}} +  \frac{\|\vr_{i}\|^{2}}{\|\vx_{i}\|^{2}}\\
       &=& \|\mw_{k} - \mw^{\star}\|^{2}_{\text{F}} - \left\|\frac{\vx_{i}\Tra(\mw_{k} - \mw^{\star})}{\|\vx_{i}\|}\right\|^{2}_{\text{F}} + \frac{\|\vr_{i}\|^{2}}{\|\vx_{i}\|^{2}}.
   \end{eqnarray*}
   The first equality follows from the fact that $\vx_{i} \in \mathcal{N}(\mi - \mPxi)$.  The second equality follows from definition of the $\mPxi$ and the fact that it is an orthogonal projector.

   Taking the expectation on both sides with respect to the sampling randomness at the $(k+1)^{th}$ iterate after conditioning on the $k^{th}$ iterate gives
   \begin{eqnarray*}
       \Exp\left[ \|\mw_{k+1} - \mw^{\star}\|^{2}_{\text{F}} \given \mw_{k}\right]
       &=& \|\mw_{k} - \mw^{\star}\|^{2}_{\text{F}} - \frac{1}{\|\mx\|^{2}_{\text{F}}}\left\|\mx(\mw_{k} - \mw^{\star})\right\|^{2}_{\text{F}} + \frac{\|\mr\|^{2}_{\text{F}}}{\|\mx\|^{2}_{\text{F}}}. 
   \end{eqnarray*}
    We employ the fact that $\|\mx(\mw_{k} - \mw^{\star})\|^{2}_{\text{F}} \ge \{\sigmaminplus(\mx)\}^{2}\|(\mw_{k}-\mw^{\star})\|^{2}_{\text{F}}$ and combine terms to obtain
    \begin{eqnarray*}
       \Exp\left[ \|\mw_{k+1} - \mw^{\star}\|^{2}_{\text{F}} \given \mw_{k}\right]
       &\le& \left(1 - \frac{\{\sigma^{+}_{\min}(\mx)\}^{2}}{\|\mx\|^{2}_{\text{F}}}\right)\left\|\mw_{k} - \mw^{\star}\right\|^{2}_{\text{F}} + \frac{\|\mr^{\star}\|^{2}_{\text{F}}}{\|\mx\|^{2}_{\text{F}}}.
   \end{eqnarray*}
   Applying the Law of Total Expectations gives
   \begin{eqnarray*}
       \Exp\left[ \|\mw_{k+1} - \mw^{\star}\|^{2}_{\text{F}} \right]
       &\le& \left(1 - \frac{\{\sigma^{+}_{\min}(\mx)\}^{2}}{\|\mx\|^{2}_{\text{F}}}\right)\Exp\left[ \left\|\mw_{k} - \mw^{\star}\right\|^{2}_{\text{F}}\right] + \frac{\|\mr^{\star}\|^{2}_{\text{F}}}{\|\mx\|^{2}_{\text{F}}}.
   \end{eqnarray*}
   Let $\kappa(\mx) = \frac{\|\mx\|^{2}_{\text{F}}}{\{\sigma^{+}_{\min}(\mx)\}^{2}}$ and observe that $1 - \frac{1}{\kappa(\mx)} < 1$.    
   Recursively applying the above bound over the previous $k$ iterations gives us
   \begin{eqnarray*}
       \Exp\left[ \|\mw_{k} - \mw^{\star}\|^{2}_{\text{F}} \right]
       &\le& \left(1-\frac{1}{\kappa(\mx)}\right)^{k}\left\|\mw_{0} - \mw^{\star}\right\|^{2}_{\text{F}} + \underbrace{\frac{1}{\{\sigma^{+}_{\min}(\mx)\}^{2}}}_{\beta} \|\mr^{\star}\|^{2}_{\text{F}}.
   \end{eqnarray*}
\end{proof}

Notice that in the consistent case, $\|\mr^{\star}\|_{F} = 0$ so we recover the exponential convergence in \cite[Theorem 2]{strohmer2009randomized} with the difference that the solution is a matrix and $\kappa(\mx)$ involves $\sigmaminplus(\mx)$ since $\mx$ is high-dimensional.  In the inconsistent case, we compare the second term on the right in Proposition \ref{prop:W-inconsistent} to the square of the second term on the right in \cite[Theorem 2.1]{needell2010randomized}, which deals with the consistent case that is observed with noise.  Using the current notation, that term in \cite[Theorem 2.1]{needell2010randomized} is $\frac{1}{\sigma_{\min}(\mx)^{2}}\|\mx\|^{2}_{\text{F}}\max_{i}\frac{|r_{i}|^{2}}{\|\vx_{i}\|^{2}}$.  The corresponding term in Proposition \ref{prop:W-inconsistent} is $\frac{1}{\{\sigmaminplus(\mx)\}^{2}}\|\mr\|^{2}_{\text{F}}$.  
Recall that \cite[Theorem 2.1]{needell2010randomized} deals with a vector right-hand side while Proposition \ref{prop:W-inconsistent} has a matrix right-hand side.  Consequently, Proposition \ref{prop:W-inconsistent} involves a matrix of residuals $\mr$.  We also compare with the second term on the right in \cite[Corollary 5.1]{needell2016stochastic}, which allows for fixing a step-size $c<1$ in the RK updates.  Using the current notation, if $c=\frac{1}{2}$, that term in \cite[Corollary 5.1]{needell2016stochastic} is $\frac{1}{\sigmaminplus(\mx\Tra\mx)}\cdot \frac{\sigma^{2}}{n\, \inf_{i}\|\vx_{i}\|^{2}}$, where $\sigma^{2} = n\sum_{i}\|\vx_{i}\|^{2}|\langle \vx_{i}, \vw^{\star}\rangle - y_{i}\|^{2}$.  Again, we highlight that Proposition \ref{prop:W-inconsistent} involves a matrix of residuals since \verb|RRLDA-RK| has a matrix right-hand side.

The second term on the right of Proposition \ref{prop:W-inconsistent} makes clear how inconsistency affects convergence.  If the system is consistent, the second term on the right-hand side vanishes.  If the system is inconsistent and $\kappa(\mx) > 1$, then as $k \rightarrow \infty$, $\left(1-\frac{1}{\kappa(\mx)}\right)^{k} \rightarrow 0$ so one cannot do better than getting within a radius of $\beta \|\mr\|^{2}_{\text{F}}$ of $\mw^{\star}$ with Algorithm \ref{alg:rkRRLDA} and row-based weights $p_{i} = \frac{\|\vx_{i}\|^{2}}{\|\mx\|^{2}_{\text{F}}}$.  Additionally, in the case that $n$ and $d$ are both very large and the system in \eqref{eqn:ls} is inconsistent, $\|\mr\|^{2}_{\text{F}}$ may be large.  Similarly, if $\sigma^{+}_{\min}(\mx)$ is very small, then the radius coefficient $\beta$ can likewise be large.

Without loss of generality, we consider the worst case scenario, where the system in \eqref{eqn:ls} may be inconsistent.  In that case, we cannot expect to do better than achieving a radius $\beta \,\|\mr\|^{2}_{\text{F}}$ of $\mw^{\star}$.  Let $\epsilon>0$ be a small tolerance.  We next consider the number of iterations $k$ in Algorithm \ref{alg:rkRRLDA} required to obtain an expected error in Proposition \ref{prop:W-inconsistent} that is less than $\beta \, \|\mr\|^{2}_{\text{F}} + \epsilon$ with row-based weights.
 
\begin{corollary}
\label{cor:numberk}
    Let $\kappa(\mx) = \frac{\|\mx\|^{2}_{\text{F}}}{\{\sigma^{+}_{\min}(\mx)\}^{2}}$ be a scaled condition number, and let $\mr = \my - \mx\mw^{\star}$ denote the residual matrix at the least norm solution $\mw^{\star}$ of \eqref{eqn:ls}.  Define $\varepsilon_{0} = \|\mw_{0} - \mw^{\star}\|^{2}_{\text{F}}$ to be the initial error.  Given a tolerance $\epsilon > 0$ and row-based weights $p_{i} = \frac{\|\vx_{i}\|^{2}}{\|\mx\|^{2}_{\text{F}}}$ for $1 \le i \le n$,
    \begin{eqnarray*}
        k \ge \frac{\log \epsilon - \log \varepsilon_{0}}{\log(1 - \frac{1}{\kappa(\mx)})}
    \end{eqnarray*}
    iterations in Algorithm \ref{alg:rkRRLDA} are sufficient to ensure that $\Exp[\|\mw_{k} - \mw_{\star}\|^{2}_{\text{F}}]  < \beta \|\mr\|^{2}_{\text{F}} + \epsilon$.
\end{corollary}

Corollary \ref{cor:numberk} gives the number of iterations $k$ required to attain a very small error $\epsilon$ of the possible floor.  Recall that in the consistent case, $\|\mr\|^{2}_{F} =0$ so the right-hand side is simply $\epsilon$.  In practice, however, the RK method exhibits much faster convergence towards the least norm solution in its initial iterates \cite{jiao2017preasymptotic, jin2018regularizing, steinerberger2021randomized, steinerberger2020regularization}.  This phenomenon is referred to as \emph{implicit regularization} since the algorithm makes greater progress along certain directions even without explicitly incorporating regularization penalties.  In the next section, we describe how Algorithm \ref{alg:rkRRLDA} imparts implicit regularization.

\subsection{Implicit regularization via the randomized Kaczmarz method}

The regularization effect of the RK method for solving invertible and overdetermined linear systems with a unique solution and a single right-hand side is explored in \cite[Theorem 1]{steinerberger2020regularization}.  Using the current notation, the left-hand side of \cite[Theorem 1]{steinerberger2020regularization} is $\Exp[\|\mx\mw_{k+1}-\my\|^{2}_{F} \given \mw_{k}]$.  Rather than looking at the expected squared norm of the $k^{th}$ RK residual, Proposition \ref{prop:implicitreg} below directly shows how the implicit regularization occurs in expectation on the RK iterates.

\begin{proposition}\label{prop:implicitreg}
    Let $\mx \in \mathbb{R}^{n \times d}$ and $\my \in \mathbb{R}^{n \times g}$ and consider solving $\min_{\mw \in \mathbb{R}^{d \times g}} \|\mx\mw - \my\|^{2}_{\text{F}}$ in \eqref{eqn:ls} via the RK as described in Algorithm \ref{alg:rkRRLDA} with $p_{i} = \frac{\|\vx_{i}\|^{2}}{\|\mx\|^{2}_{\text{F}}}$.  Let $\mw^{\star} = \mx\Dag\my$ denote the least norm solution to \eqref{eqn:ls} and let $\Exp[ \, \cdot \,]$ abbreviate $\Exp_{i}[ \, \cdot \,| \,\mw_{k}]$, the expectation with respect to the sampling randomness at the $(k+1)^{th}$ iterate after conditioning on the $k^{th}$ iterate.  Then 
    \begin{eqnarray*}
        \Exp[\mw_{k+1} - \mw_{k}] &=& -\frac{1}{\|\mx\|^{2}_{F}}\sum_{j=1}^{r}\sigma_{j}^{2}\vv_{j}\vgamma_{j}\Tra,
    \end{eqnarray*}
    where the $\vv_{j}$ for $j \in [r]$ form a basis for $\mathcal{R}(\mx\Tra)$ and the $\vgamma_{j} \in \mathbb{R}^{g}$ are the corresponding coefficient vectors for the basis expansion of the matrix left-hand side.
\end{proposition}

\begin{proof}
   Beginning with the RK update in Line 5 of Algorithm \ref{alg:rkRRLDA}, we have
   \begin{eqnarray*}
       \mw_{k+1} - \mw_{k} &=& \frac{\vx_{i}}{\|\vx_{i}\|^{2}_{2}}(\vy_{i}\Tra - \vx_{i}\Tra\mw_{k}).
   \end{eqnarray*}
   Taking expectations with respect to the sampling randomness on both sides after conditioning on the $k^{th}$ iterate gives
   \begin{eqnarray*}
       \Exp[\mw_{k+1} - \mw_{k}] &=& -\frac{1}{\|\mx\|^{2}_{F}}\mx\Tra(\mx\mw_{k} - \my).
   \end{eqnarray*}
   Since $\mw^{\star}$ satisfies the normal equations, $\mx\Tra\mx\mw^{\star} = \mx\Tra \my$ so that
   \begin{eqnarray}\label{eqn:implicitreg_1}
       \Exp[\mw_{k+1} - \mw_{k}] &=& -\frac{1}{\|\mx\|^{2}_{F}}\mx\Tra\mx(\mw_{k} - \mw^{\star}).
   \end{eqnarray}
   Let ${\mx = \mU\md\mv\Tra}$ be the economic singular value decomposition of $\mx$ with $\rank(\mx) = r \le \min\{n, d\}$ so that $\mv \in \mathbb{R}^{d \times r}$.  Then $\mv\mv\Tra$ is the orthogonal projector onto $\mathcal{R}(\mx\Tra)$.  Let $\widetilde{\mv} \in \mathbb{R}^{d \times (d-r)}$ be the orthogonal complement of $\mv$ so that $\widetilde{\mv}\widetilde{\mv}\Tra$ is the orthogonal projector onto $\mathcal{N}(\mx)$ and $\mv\mv\Tra + \widetilde{\mv}\widetilde{\mv}\Tra = \mi_{d}$.  Therefore, we can write $\mw_{k} - \mw^{\star}$ with a basis expansion as
    \begin{eqnarray*}
        \mw_{k} - \mw^{\star} &=& \sum_{a=1}^{r} \vv_{a} \vgamma_{a}\Tra + \sum_{b=1}^{d-r} {\tilde{\vv}}_{b}\vtheta_{b}\Tra,
    \end{eqnarray*}
    where $\vv_{a}$ for $a \in [r]$ are the columns of $\mv$, ${\tilde{\vv}}_{b}$ for $b \in [d-r]$ are the columns of $\widetilde{\mv}$, and the $\vgamma_{a}$ and $\vtheta_{b}$ are the corresponding coefficient vectors in $\mathbb{R}^{g}$ that express this basis expansion for a matrix.  Therefore, left multiplying $\mw_{k} - \mw^{\star}$ by $\mx\Tra\mx$ gives
    \begin{eqnarray*}
        \mx\Tra\mx(\mw_{k} - \mw^{\star}) &=& \mx\Tra\mx\sum_{a=1}^{r} \vv_{a} \vgamma_{a}\Tra \\
        &=& \sum_{j=1}^{r}\sigma_{j}^{2}\vv_{j}\vgamma_{j}\Tra,
    \end{eqnarray*}
    where the $\sigma_{j}$ for $1\le j\le r$ are the ordered non-zero singular values of $\mx$.  The first equality follows from the fact that $\widetilde{\mv}$ is a basis for $\mathcal{N}(\mx)$.  The second equality follows from the fact that the $\vv_{j}$ are orthogonal.  Inserting the last equality in \eqref{eqn:implicitreg_1} gives
    the result.
\end{proof}

Now we see how Algorithm \ref{alg:rkRRLDA} imparts implicit regularization.  Proposition \ref{prop:implicitreg} says that at each iterate, the RK iterates $\mw_{k}$ stay in the row space of $\mx$.  Moreover, progress from one iterate to the next is most heavily weighted (by $\sigma_{j}^{2}$) towards subtracting directions that correspond to the largest singular vectors of $\mathcal{R}(\mx\Tra)$ remaining in $\mw_{k} - \mw^{\star}$. Employing the same economic SVD of $\mx$ from the proof of Proposition \ref{prop:implicitreg} above, recall that the least norm solution
\begin{eqnarray*}
    \mw^{\star} &=& \mx\Dag\my \quad=\quad \sum_{j=1}^{r}\frac{1}{\sigma_{j}}\vv_{j}\vu_{j}\Tra\my,
\end{eqnarray*}
where the $\sigma_{j}$ are the ordered singular values of $\mx$.  Therefore, the directions corresponding to singular vectors with the largest singular values of $\mx$ are down-weighted in the least norm solution.  Indeed, Proposition \ref{prop:implicitreg} shows that at each iterate, the RK update removes the directions corresponding to the singular vectors of $\mx$ remaining in $\mw_{k} - \mw^{\star}$ with largest singular values.  Therefore, Algorithm \ref{alg:rkRRLDA} steers the updates towards the least norm solution even without explicit regularization.  

While Proposition \ref{prop:implicitreg} holds generally for linear systems with matrix right-hand sides, it certainly also holds for systems with a single right-hand side.  In that case, $g=1$ and the $\vgamma_{j} \in \mathbb{R}^{g}$ are simply scalars.

\subsection{Computational Complexity of RRLDA-RK}

Assuming that $\mx$ is a general dense matrix and that $g \ll d$, each iteration of Algorithm \ref{alg:rkRRLDA} requires $\mathcal{O}(d)$ flops.  If there are a maximum of $K$ iterations, then the computational cost of running Algorithm \ref{alg:rkRRLDA} is $\mathcal{O}(Kd)$ flops.  While the $K$ required to obtain a very small error tolerance of the floor as detailed in Corollary \ref{cor:numberk} can be quite large, we observe that $K$ can be much smaller in practice when the goal is a downstream machine learning tasks.  This was observed in \cite{chi2025linear} and also exhibited in the numerical experiments below.

\section{Numerical Experiments}
\label{sec:experiments}

In this section, we present numerical experiments on real datasets to highlight the efficacy of the \verb|RRLDA-RK| subspace from Algorithm \ref{alg:rkRRLDA} compared with previously described alternatives for downstream classification tasks.  We describe the experimental setup, the real datasets, and additional details relevant to each dataset below.

\subsection{Experimental Setup}

We compare the methods previously discussed for computing a reduced dimension subspace prior to downstream classification.  These include \verb|LSLDA| from \cite{ye2007least}, \verb|ULDA| from \cite[Algorithm 1]{ye2005characterization}, \verb|LN Sol| from the least norm solution to \eqref{eqn:ls} obtained via LSQR \cite{lsqr}, and \verb|RRLDA-RK| from Algorithm \ref{alg:rkRRLDA}.  Since our data are large and high-dimensional so that the system in \eqref{eqn:ls} is underdetermined, we employ an operator-based implementation of the LSQR algorithm with the \verb|scipy.sparse.linalg.LinearOperator| class to define $\mx$ as a linear operator through functions that evaluate matrix–vector products with $\mx$ and $\mx\Tra$.  Since the data are high-dimensional, we do not compute the $\mg^{\star}$ from classical RRLDA in \eqref{eqn:ulda1} since $(\ms_{t}^{G})\Inv$ does not exist.  As a baseline for comparison, we additionally include \verb|Full|, which denotes the full data without any dimension reduction.  This results in five comparison methods for dimension reduction for each dataset and classification method: \verb|Full|, \verb|LSLDA|, \verb|ULDA|, \verb|LN Sol|, and \verb|RRLDA-RK|.    

For each of the five comparison methods and for each dataset, we draw 30 random splits of 70 percent training and 30 percent testing observations.  For each replicate, when possible, we first compute a reduced dimension subspace from the training data.  Then, we right multiply the testing subset of the design matrix $\mx$ with each of the reduced dimension subspaces described above.  Finally, we apply several classification methods, described below, to the resulting spaces.  We compare the classification accuracy and compute time required for each of the comparison methods and report the median and standard deviation for each.

We perform all accuracy and timing experiments on CPU compute nodes on the Agate high-performance computing cluster operated by the Minnesota Supercomputing Institute at the University of Minnesota.  All implementations are in Python and rely on optimized numerical linear algebra libraries for efficient use of large matrix–vector operations and iterative solvers.

\subsection{Twenty Newsgroups Dataset}

Twenty Newsgroups refers to sparse text classification data derived from $18{,}846$ news articles, each belonging to one of 20 news topics \cite{twentynewsgroups}.  The features are a term frequency-inverse document frequency (TF-IDF) representation of each article's text \cite{aizawa2003information, ramos2003using} obtained using \verb|TfidfVectorizer| in the \verb|scikit-learn| Python libraries \cite{scikit-learn}.  We employ the maximum number of features available for the TF-IDF representation, resulting in a dataset with $n=18{,}846$ observations and $p=130{,}000$ features.

Following the experimental setup in \cite{ye2007least}, we employ $k$-nearest neighbors (kNN) to assess how the learned subspaces perform on downstream classification tasks.  While \cite{ye2007least} considers $k=1$ neighbor, we additionally run experiments with $k=5$ and $k=10$ neighbors to evaluate the sensitivity of accuracy to the number of neighbors.  We employ kNN with the Euclidean norm via \verb|KNeighborsClassifier| in the \verb|scikit-learn| Python libraries \cite{scikit-learn}.

Figure \ref{fig:acc_time_20newsgroups} summarizes classification accuracy (left panel) and wall-clock compute time (right panel) across the three kNN classifiers (kNN-1, kNN-5, and kNN-10) for each comparison subspace construction method.  Since \verb|ULDA| was not able to complete any replicates due to memory issues on this dataset, we omit it from the figure.  Across all values of $k$, the baseline \verb|Full| with no dimension reduction yields the lowest accuracy and its performance decreases slightly as $k$ increases.  By contrast, all three reduced dimension representations substantially improve accuracy with \verb|LSLDA| and \verb|LN Sol| producing identical results and the highest overall performance.  \verb|RRLDA-RK| achieves accuracy that is comparable to these two methods and does so with less variance, as indicated by the smaller standard deviation bars.  While accuracy with \verb|RRLDA-RK| is typically slightly below \verb|LSLDA| and \verb|LN Sol|, it is consistently above \verb|Full| on this dataset.

The timing results reported in the right panel of Figure \ref{fig:acc_time_20newsgroups} on a logarithmic scale highlight the substantial computational differences.  Applying kNN on the full data (\verb|Full|) is extremely fast since there is no reduced dimension subspace to compute but less accurate.  While $\verb|LSLDA|$ obtains the same solution as \verb|LN Sol| on this dataset, it requires the most computational time.  This reflects the expense of employing SVD-based solvers for large-scale least-squares problems.  Meanwhile, \verb|LN Sol| is slightly faster than \verb|LSLDA| but still requires more computational time than \verb|RRLDA-RK|.  By contrast, \verb|RRLDA-RK| achieves one to two orders of magnitude speedup over \verb|LSLDA| and \verb|LN Sol|, with compute times on the order of tens of seconds, while attaining very comparable classification accuracy. 

\begin{figure}[H]
  \centering
  \includegraphics[width=\linewidth]{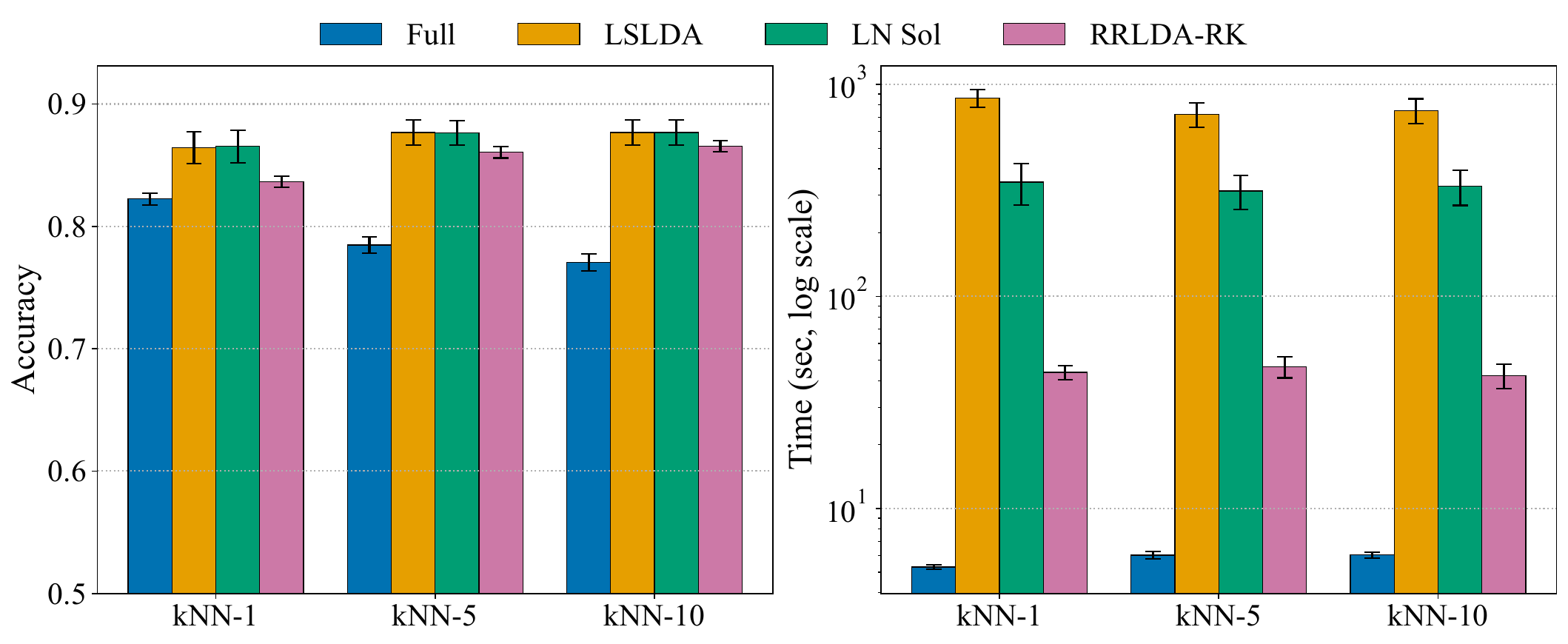}
  \caption{Accuracy (left) and timing (right) across kNN settings (knn1, knn5, knn10) for all methods. Bars indicate median accuracy and time (in log seconds); error bars indicate standard deviations across runs.}
  \label{fig:acc_time_20newsgroups}
\end{figure}

To understand why \verb|RRLDA-RK| does not obtain the same result as \verb|LNSOL| in this case, we investigate whether the training replicates drawn from the 20 Newsgroups dataset result in consistent systems in \eqref{eqn:ls}.  To do this, we obtain 10 random draws of 70 percent training data from this dataset and compute $\mw^{\star}$ with \verb|LNSOL| on each subset.  We then compute the residual $\mr^{\star}$ at $\mw^{\star}$ using $\my$ from the training data and as described in Algorithm \ref{alg:rkRRLDA}.  The minimum, median, and maximum values of $\|\mr^{\star}\|_{\text{F}}$ were 48.04, 55.5, and 60.74, respectively, with a standard deviation of 3.89.  While none of the draws resulted in consistent linear systems, the overall residual is also relatively small since the Frobenius norm was computed over $n \times p$ entries.  The relative value of $\|\mr^{\star}\|_{\text{F}}$, computed as $\frac{\|\mr^{\star}\|_{\text{F}}}{\|\my\|_{\text{F}}}$, was 0.11 with a standard deviation of 0.01.  As indicated by the second term on the right side of Proposition \ref{prop:W-inconsistent}, the nonzero residuals explain the comparable, but slightly lower, accuracy from \verb|RRLDA-RK| compared with \verb|LNSOL| on this dataset.

\subsection{OASIS Magnetic Resonance Imaging (MRI) Dataset}

The Open Access Series of Imaging Studies (OASIS-1) is a publicly available cross-sectional dataset consisting of T1-weighted MRI scans from $416$ subjects between the ages of 18 and 96 \cite{marcus2007open}.  Of the $416$ subjects, $100$ of those over the age of $60$ have a clinical diagnosis of mild or moderate Alzheimer's disease (AD).  An additional reliability set of $20$ MRI scans from non-demented patients in the dataset brings the total number of subjects to $436$ \cite{marcus2007open}.

For each subject, we select a single atlas-aligned, brain-masked, and intensity-corrected volume (indicated by the \_111\_t88\_masked\_gfc processing tag) for consistent pre-processing and improved voxel-wise comparability.  We resample all MRI volumes to a common grid with the \verb|NiBabel| \cite{nibabel} and \verb|Nilearn| \cite{nilearn} Python libraries, standardize the volumes within subject, and vectorize the resulting volumes to produce a subject $\times$ voxel design matrix.  This results in a design matrix with $n=436$ subjects and $p=1{,}741{,}740$ voxels.  We obtain diagnostic class labels from the Clinical Dementia Rating (CDR) from the demographic and clinical data to obtain two classes: those with mild or moderate dementia or AD (CDR $> 0$) and those without it (CDR $= 0$).

To assess how the learned subspaces perform on downstream classification tasks, we compare classification accuracy with several standard linear and nonlinear classifiers in the \verb|scikit-learn| Python libraries \cite{scikit-learn}. These include $k$-nearest neighbors with $k=10$ via \verb|KNeighborsClassifier|, linear support vector machines (SVM) via \verb|LinearSVC|, and logistic regression with $\ell_{2}$- and $\ell_{1}$-regularization (LR-l2 and LR-l1) via \verb|LogisticRegression| with $C=1.0$ and $\alpha=1e-4$, respectively.

Table \ref{tab:acc_oasis} reports classification accuracy across the four downstream classifiers for each comparison subspace construction method on the OASIS-1 dataset.  Accuracy with the \verb|LN Sol| subspace is consistently the highest while \verb|RRLDA-RK| and \verb|Full| produce very comparable results in general.  Meanwhile, accuracy obtained with the subspaces from \verb|LSLDA| and \verb|ULDA| is substantially worse in nearly all instances.

\begin{table}[H]
\caption{Accuracy results on OASIS-1 dataset.  Columns depict median classification accuracy ($\pm$ standard deviation) with subspaces obtained from comparison methods.  Rows indicate classifiers.  Highest accuracy in each row indicated in boldface.}
\label{tab:acc_oasis}
\centering
\resizebox{\columnwidth}{!}{%
\begin{tabular}{lrrrrr}
\toprule
 & Full & LSLDA & LN Sol & ULDA & RRLDA-RK \\
Classifier &  &  &  &  &  \\
\midrule
kNN-10 & 0.81 $\pm$ 0.03 & 0.54 $\pm$ 0.10 & \textbf{0.83} $\pm$ 0.02 & 0.77 $\pm$ 0.12 & 0.82 $\pm$ 0.03 \\
SVM & 0.82 $\pm$ 0.02 & 0.78 $\pm$ 0.23 & \textbf{0.83} $\pm$ 0.02 & 0.77 $\pm$ 0.24 & \textbf{0.83} $\pm$ 0.03 \\
LR-l2 & \textbf{0.82} $\pm$ 0.02 & 0.79 $\pm$ 0.03 & \textbf{0.82} $\pm$ 0.02 & \textbf{0.82} $\pm$ 0.03 & \textbf{0.82} $\pm$ 0.03 \\
LR-l1 & \textbf{0.81} $\pm$ 0.03 & 0.78 $\pm$ 0.03 & \textbf{0.81} $\pm$ 0.02 & 0.78 $\pm$ 0.03 & \textbf{0.81} $\pm$ 0.03 \\
\bottomrule
\end{tabular} }
\end{table}

Table \ref{tab:time_oasis} reports the corresponding compute times in seconds.  Although classification with the \verb|LN Sol| subspace produced the highest classification accuracy overall, computing this subspace also required the most compute time by far.  In fact, the time required to compute and classify the \verb|LN Sol| subspace takes approximately 45 times as long as the time required to compute and classify the \verb|RRLDA-RK| subspace.  By contrast, the \verb|RRLDA-RK| subspace achieved as high, or nearly as high, classification accuracy as \verb|LN Sol| in all instances and did so at a fraction of the compute time.  Note that classification with \verb|Full| also achieved accuracy comparable to \verb|RRLDA-RK|.  Nonetheless, in three of the four classification scenarios, it required substantially more compute time than \verb|RRLDA-RK|.

\begin{table}[H]
\caption{Time results on Oasis-1 dataset.  Columns depict median compute time in seconds ($\pm$ standard deviation) for each comparison method.  Rows indicate classifiers.  Highest accuracy in each row indicated in boldface.}
\label{tab:time_oasis}
\centering
\resizebox{\columnwidth}{!}{%
\begin{tabular}{lrrrrr}
\toprule
 & Full & LSLDA & LN Sol & ULDA & RRLDA-RK \\
Classifier &  &  &  &  &  \\
\midrule
kNN-10 & 115.72 $\pm$ 2.40 & 14.11 $\pm$ 0.50 & 453.23 $\pm$ 117.35 & 26.59 $\pm$ 1.34 & \textbf{9.95} $\pm$ 0.84 \\
SVM & 131.75 $\pm$ 210.72 & 14.07 $\pm$ 0.48 & 453.44 $\pm$ 117.34 & 26.49 $\pm$ 1.41 & \textbf{9.95} $\pm$ 0.84 \\
LR-l2 & 54.55 $\pm$ 4.71 & 14.03 $\pm$ 0.48 & 453.43 $\pm$ 117.33 & 26.52 $\pm$ 1.39 & \textbf{9.94} $\pm$ 0.85 \\
LR-l1 & 17.81 $\pm$ 0.67 & 14.07 $\pm$ 1.03 & 453.37 $\pm$ 117.36 & 27.08 $\pm$ 2.41 & \textbf{9.96} $\pm$ 0.85 \\
\bottomrule
\end{tabular} }
\end{table} 

\section{Conclusion}
\label{sec:conclusion}

In this work, we introduce a scalable randomized algorithm for performing classical and high-dimesional RRLDA on large-scale data via the randomized Kaczmarz method.  We also provide convergence analysis for the algorithm and intuition for how the algorithm imparts implicit regularization towards the least norm solution without explicitly incorporating penalties.  Our numerical experiments on real data demonstrate that \verb|RRLDA-RK| in Algorithm \ref{alg:rkRRLDA} can offer a favorable accuracy vs. computational time trade-off.  Classification with the \verb|RRLDA-RK| subspace produces accuracy that is as good, or nearly as good, as state-of-the-art approaches with potentially very large computational time savings.  This makes it particularly suited for large-scale, high-dimensional classification problems.

We employ the RK as introduced in \cite{strohmer2009randomized} and as discussed in \cite{needell2010randomized, ma2015regs}.  There are now many was to accelerate the RK and to improve on the solution in the inconsistent case.  We mention some approaches as potential avenues of future work.  These include changing the step-size \cite{censor1983strong, hanke1990acceleration, needell2013two, tanabe1971projection, whitney1967two}, sampling blocks of rows from $\mx$ instead of a single row \cite{needell2014paved, rebrova2020block, necoara2019faster}, leveraging larger singular values \cite{derezinski2025randomized}, and parallel implementations \cite{liu2014asynchronous}.  These methods could be incorporated to accelerate Algorithm \ref{alg:rkRRLDA}, making this approach particularly attractive for extremely large-scale, high-dimensional data.


\bibliographystyle{elsarticle-num} 
\bibliography{randmulticlasslda}



\end{document}